\documentclass[11pt]{fundam}

\usepackage{amssymb}
\usepackage[toc,page]{appendix}

\usepackage{amssymb}
\usepackage{graphicx,multirow}
\usepackage{url}
\usepackage{algorithmic}
\usepackage{qtree}
\usepackage[ruled,vlined]{algorithm2e}

\usepackage{color}

\newtheorem{property}{Property}

\begin{document}



\title{Prefix Block-Interchanges on Binary and Ternary Strings}

\address{msrahman{@}cse.buet.ac.bd}

\author{Md. Khaledur Rahman\\
Bangladesh University of Engineering and Technology, Dhaka-1000, Bangladesh\\
khaled.cse.07@gmail.com
\and M. Sohel Rahman\\
Department of CSE \\
Bangladesh University of Engineering and Technology, Dhaka-1000, Bangladesh\\
msrahman{@}cse.buet.ac.bd
 }

\runninghead{M. K. Rahman, M. S. Rahman}{Prefix Block-Interchanges on Binary and Ternary Strings}
\maketitle
\begin{abstract}
The genome rearrangement problem computes the minimum number of operations that are required to sort all elements of a permutation. A block-interchange operation exchanges two blocks of a permutation which are not necessarily adjacent and in a prefix block-interchange, one block is always the prefix of that permutation. In this paper, we focus on applying prefix block-interchanges on binary and ternary strings. We present upper bounds to group and sort a given binary/ternary string. We also provide upper bounds for a different version of the block-interchange operation which we refer to as the `restricted prefix block-interchange'. We observe that our obtained upper bound for restricted prefix block-interchange operations on binary strings is better than that of other genome rearrangement operations to group fully normalized binary strings. Consequently, we provide a linear-time algorithm to solve the problem of grouping binary normalized strings by restricted prefix block-interchanges. We also provide a polynomial time algorithm to group normalized ternary strings by prefix block-interchange operations. Finally, we provide a classification for ternary strings based on the required number of prefix block-interchange operations.

\end{abstract}


\section{Introduction}

Genome rearrangement has become a popular combinatorial problem in the field of computational biology which helps to trace the evolutionary distance between two species. The genome is structurally specific to each species, and it changes only slowly over time. Though the gene content of two genomes is almost identical, gene order can be quite different, which results in different species \cite{sankoff}. There are several genome rearrangement events such as \emph{reversal, transposition, block-interchanges, etc.} \cite{christie,hartman,christie2,rahman4} that contribute much to such biological diversity. In this paper, we will only discuss block-interchanges and its variations. 

In computational biology, a genome is often represented as a \emph{permutation} consisting of all elements from 1 to $n$, where $n$ is the length of the \emph{permutation}. In an \emph{identity} permutation, all elements are sorted in ascending order from 1 to $n$. A block-interchange involves swapping two blocks of a permutation which are not necessarily adjacent. A prefix block-interchange is a variation of block-interchange where one of the two blocks must be a prefix of the permutation. The prefix block-interchange distance between two permutations can be used to estimate the number of global mutations between genomes and can be used by molecular biologists to infer evolutionary and functional relationships. \emph{Sorting by prefix block-interchanges} is the problem of finding the minimum number of prefix block-interchange operations needed to transform a given permutation into the \emph{identity} permutation. 

A natural variant of the previously mentioned sorting problems is to apply genome rearrangement operations not on permutations but on strings over fixed size alphabets. A string over the alphabet $\Sigma$ is a sequence of symbols where repetitions of symbols in $\Sigma$ are allowed. This shift is inspired by the biological observation that multiple ``copies" of the same gene can appear at various places along the genome \cite{sankoff}. 

The grouping problem of normalized strings by genome rearrangement operations provides upper bound i.e., the number of required operations to group all symbols of a string and the sorting problem computes the number of operations required to sort all grouped symbols of that normalized string by genome rearrangement operations. 
Indeed, some interesting works by Christie and Irving~\cite{christie}, Radcliffe et al.~\cite{radcliffe}, Hurkens et al.~\cite{hurkens}, Dutta et al.~\cite{amit}, and Rahman and Rahman~\cite{khaled} have explored the consequences of switching from permutations to strings. Notably, such rearrangement operations on strings have been found to be interesting and important in the study of orthologous gene assignment~\cite{chen}, especially if the strings have only low level of symbol repetition. We provide a summary of results on grouping and sorting strings by genome rearrangement operations in Table \ref{summary-tab}.

\begin{table}[]
\centering

\begin{tabular}{|l|c|l|l|c|c|}
\hline
\textbf{Authors}                                                                              & \multicolumn{2}{l|}{\textbf{GR operation}}                                                                                & \textbf{Type} & \textbf{Grouping} & \textbf{Sorting} \\ \hline
\multirow{2}{*}{Hurkens et al.~\cite{hurkens}}                                              & \multicolumn{2}{c|}{\multirow{2}{*}{PR}}                                                            & Binary          &         $n-2$    & $n-2$\\ \cline{4-6} 
                                                                                     & \multicolumn{2}{l|}{}                                                                                            & Ternary         &        $n-3$  &   $n-3$\\ \hline
\multirow{2}{*}{Dutta et al.~\cite{amit}}                                                   & \multicolumn{2}{c|}{\multirow{2}{*}{PT}}                                                       & Binary          &       $\lceil \frac{n-2}{2} \rceil$     & $\lceil \frac{n-2}{2} \rceil$ \\ \cline{4-6} 
                                                                                     & \multicolumn{2}{c|}{}                                                                                            & Ternary         &      $\lceil \frac{n-3}{2}  \rceil + 1$   &  $\lceil \frac{n-3}{2}  \rceil + 2$ \\ \hline
\multirow{2}{*}{\begin{tabular}[c]{@{}l@{}}Rahman and \\ Rahman~\cite{khaled}\end{tabular}} & \multicolumn{2}{c|}{\multirow{2}{*}{\begin{tabular}[c]{@{}c@{}}PSTR\end{tabular}}} & Binary        & $\lceil \frac{n-2}{2} \rceil $  &  $\lceil \frac{n-2}{2} \rceil +1$           \\ \cline{4-6} 
                                                                                     & \multicolumn{2}{l|}{}                                                                                            & Ternary     & $\lceil \frac{n-2}{2} \rceil + 1$    &            $\lceil \frac{n-2}{2} \rceil + 2$ \\ \hline
Chou et al.~\cite{compete}                                                                  & \multicolumn{2}{c|}{PBI}                                                                    & Binary          &         $\lceil \frac{n-2}{3} \rceil$  &  $\lceil \frac{n-2}{3} \rceil$  \\ \hline
\multirow{3}{*}{\begin{tabular}[c]{@{}l@{}}This article \end{tabular}} & \multicolumn{2}{c|}{\multirow{1}{*}{\begin{tabular}[c]{@{}l@{}}RPBI\end{tabular}}} & Binary     & $\lceil \frac{n-2}{4} \rceil$     &  X           \\ \cline{4-6}
& \multicolumn{2}{c|}{\multirow{1}{*}{\begin{tabular}[c]{@{}c@{}}RPBI\end{tabular}}} & Ternary     & $\lceil \frac{n-2}{2} \rceil$     &  X           \\ \cline{4-6}
& \multicolumn{2}{c|}{\multirow{1}{*}{\begin{tabular}[c]{@{}c@{}}PBI\end{tabular}}} & Binary   & $\lceil \frac{n-2}{3} \rceil$       &  $\lceil \frac{n-2}{3} \rceil +1$           \\ \cline{4-6} 
                                                                                     & \multicolumn{2}{c|}{PBI}                                                                                            & Ternary         &         $\lceil \frac{n-2}{2} \rceil $  & $\lceil \frac{n-2}{2} \rceil + 2$ \\ \hline
\end{tabular}
\caption{A list of results on grouping and sorting strings by Genome Rearrangement (GR) operations. Here, $n$ is the length of the string. Other notations are as follows: $X$-Not Available, $PR$-Prefix Reversal, $PT$-Prefix Transposition, $PSTR$-Prefix and Suffix TransReversal, $PBI$-Prefix Block-Interchange and $RPBI$-Restricted Prefix Block-Interchange.}
\label{summary-tab}
\end{table}

Chen et al.~\cite{chen}, presented polynomial-time algorithms for computing the minimum number of reversals and transpositions operations to sort a given binary string. They also gave exact constructive diameter results on binary strings. Radcliffe et al.~\cite{radcliffe} on the other hand gave refined and generalized reversal diameter results for non fixed size alphabets. Hurkens et al.~\cite{hurkens} introduced \emph{grouping} (a weaker form of sorting), where identical symbols need only be grouped together, while groups can be in any order. In the sequel, they gave a complete characterization of the minimum number of prefix reversals required to group (and sort) binary and ternary strings. Their proposed upper bound for grouping binary string is $n-2$, where $n$ is the length of the binary string $s$. Subsequently, Dutta et al. \cite{amit} followed up their work~\cite{hurkens} on binary and ternary strings to apply prefix transpositions introducing relabeling. They gave a complete characterization of the minimum number of prefix transpositions required to group (and sort) binary and ternary strings. Their proposed upper bound to group a binary string is $\lceil \frac{n-2}{2} \rceil$, where $n$ is the length of the string. It may be noted that, apart from being a useful aid for sorting, grouping itself is a problem of interest in its own right~\cite{eriksson}. Chou et al. also proposed a bound only for binary strings applying prefix block-interchanges \cite{compete}. Their proposed bound to group binary strings is $\lceil \frac{n-2}{3} \rceil$. They also proposed a linear time algorithm to sort binary strings. Finally, Rahman and Rahman \cite{khaled} worked on prefix and suffix versions of transreversal operations and derived similar bounds for binary and ternary strings as in \cite{amit}. Their proposed bound to group binary strings considering one operation at a time is $\lceil \frac{n-2}{2} \rceil + 1$, where $n$ is the length of the string. 

In this paper, we follow up the works of~\cite{amit,khaled,hurkens,compete} and consider prefix block-interchange operations to group and sort binary and ternary strings. Notably, as a future work in \cite{hurkens}, the authors raised the issue of considering other genome rearrangement operators. The main contributions of this paper are as follows. We deduce the number of prefix block-interchanges required to group and sort binary strings (Section \ref{section_grp}). Then we find the number of prefix block-interchanges required to group and sort ternary strings (Section \ref{section_grpt}). Our deduced bound to group fully binary normalized string $s$ is $\lceil \frac{n-2}{4} \rceil$ considering restricted prefix block-interchange operations where $n$ is the length of the string (Section \ref{restrictedpblock}). This is interesting because this bound is better than other proposed bounds for other rearrangement operations in the literature. We also provide a linear-time algorithm to group binary normalized strings by restricted prefix block-interchange operations. Then, we deduce the upper bound of $\lceil \frac{n-2}{2} \rceil $ to group fully normalized ternary strings by prefix block-interchanges, where $n$ is the length of the string. We also provide a polynomial time algorithm to group normalized ternary strings (Section \ref{section_grpt}). We also present a classification of normalized ternary strings (Section \ref{sec:class}).

\section{Preliminaries}\label{pre}
We use the similar notations and definitions used in~\cite{amit,khaled,hurkens,compete}, which are briefly reviewed below for the sake of completeness. We represent a binary string by $s = s[1]s[2]\ldots s[n]$ such that $\forall{i}: s[i] = 0$ or $1$. Similarly, we represent a ternary string by $s = s[1]s[2]\ldots s[n]$ such that $\forall{i}: s[i] = 0$ or $1$ or $2$.

We define block-interchange $\beta(w,x,y,z)$ on a string $s$ of length $n$, where $1\le w\le x < y \le z \le n$ and $s=s[1]\ldots s[w] \ldots s[x] \ldots s[y] \ldots s[z] \ldots s[n]$, as a rearrangement event that transforms ${s}$ into $s[1] \ldots s[w-1] s[y] \ldots s[z] s[x+1]\ldots s[y-1]s[w]\ldots s[x] s[z+1] \ldots s[n]$. When we consider prefix block-interchange, we actually perform $\beta(1,x,y,z)$. So, a prefix block-interchange event transforms $s$ into $s[y] \ldots s[z] s[x+1]\ldots s[y-1]s[1]\ldots s[x] s[z+1] \ldots s[n]$. The prefix block-interchange distance $d_{pbi}^s(s)$ of $s$ is defined as the minimum number of prefix block-interchanges required to sort the string $s$. Similarly, we define $d_{pbi}^g(s)$ as the minimum number of prefix block-interchanges required to group all symbols of the string $s$. After a prefix block-interchange operation, some adjacent symbols of the new formed string may be identical. We reduce all adjacent identical symbols to one symbol and get a string of reduced length which we call a \emph{normalized} string. So, in a normalized string, there will be no identical adjacent symbols. We consider two strings to be equivalent if we can convert them to the same normalized string. As representatives of the equivalence classes we take the shortest string of each class. Clearly, these are normalized strings where adjacent symbols are always different. In our work, we only consider normalized strings as in \cite{amit,khaled,hurkens}. This does not lose generality because we can convert any string to the respective normalized string, i.e., the representative of the class it belongs to.


For example, let us consider two strings `110001001011' and `100100011101' which are equivalent because these two strings have the same normalized string which is `1010101'. So, they belong to the same equivalence class where `1010101' is the representative. We consider $s=1010101$ and we want to apply a prefix block-interchange operation $\beta(1,2,5,5)$ on $s$. Now, $s[1]\ldots s[x]=s[1]s[2]=10$, $s[y]\ldots s[z]=s[5]=1$, $s[n] = s[7] = 1$. Therefore, after applying the operation we get $s=s[5]s[3]s[4]s[1]s[2]s[6]s[7]$ and finally get a reduced string as follows: $\overline{10}10\underline{1}01=\textbf{11}01\textbf{00}1 = 10101$.

\ \ \ A reduction that decreases the string length by $l$ after applying a prefix block interchange is called an $l$-pblockInterchange. So, if $l=0$, then we have a 0-pblockInterchange. The above example illustrates a 2-pblockInterchange.
%
%
%
%
%
%
%
%
%
%

A \emph{token} is a string of one or more symbols (0, 1 or 2) that is significant as a group. A \emph{quantifier} is a type of determiner that indicates quantity and it is basically applied after a \emph{token}. For instance, `+' is a \emph{quantifier} which indicates one or more occurrences of the \emph{token} after which `+' is applied. Note that a token may also consist of one symbol only. For example, consider the string $(10)^+201$. Here, \emph{10} is a \emph{token}, `1' and `0' are symbols of the token \emph{10} and `+' is the \emph{quantifier}. It means $10201$, $1010201$ or $101010201$ etc. belongs to the same representation class $(10)^+201$.

\section{Grouping and Sorting Binary Strings}\label{section_grp}
  The task of sorting a string can be divided into two subproblems, namely, \textit{grouping} the identical symbols together and then putting the groups of identical symbols in the right order. 
The possible length of binary string is either even or odd. As strings are normalized, only 4 kinds of binary strings are possible, namely, $(01)^+$, $(10)^+$, $(01)^+0$ and $(10)^+1$. We achieve the following properties for binary strings.

\begin{property}
\label{property1}
Let $s$ be a fully binary normalized string of length $n$, where $n\ge 2$. Then, we will need to reduce $n-2$ symbols to group $s$.
\end{property}

There will be at least one identical symbol in $s$ if $n>2$. So, we need to reduce those symbols applying prefix block-interchanges and then reducing consecutive identical symbols which results in a normalized string. Finally, after completion of grouping, we will get either 01 or 10, where 01 is already sorted. If we get 10, we will need one extra operation to sort it. 


\begin{property}
\label{property21}
Let $s$ be a fully binary normalized string of length $n$, where $1 < x \le y < n$ and $n\ge 5$. If we have a prefix block $s[1]$ and another block $s[x]\ldots s[y]$ such that all three of the following conditions (1-3) are satisfied, then we get a 3-pblockInterchange.
\end{property}
\begin{itemize}
\item{Condition 1: $s[1]=s[y+1]$}
\item{Condition 2: $s[1] =s[x-1]$}
\item{Condition 3: $s[x]=s[y]$} 
\end{itemize}

\begin{lemma}
\label{lemma0}
Let $s$ be a fully normalized binary string of length $n$ such that $n\ge 5$. Then, we always get 3-pblockInterchange(s) using prefix block-interchange operations.
\end{lemma}
\begin{proof}
Let $s$ be a binary normalized string, where $n = 5$. Suppose, $s=ababa$ represents a class of binary normalized strings where either $a=1$ and $b=0$, or $a=0$ and $b=1$. Now, if we exchange $s[1]$ with $s[4]$ ($\beta(1,1,4,4)$), we get a normalized string of $ba$ which is a reduced string from the original one. Here, the length of reduced string is 2 which is 3 less than the original string, i.e., we have a 3-pblockInterchange here. Now, if we append any length of binary normalized string after $s$, then according to Property \ref{property21}, we always get a 3-pblockInterchange if we always perform $\beta(1,1,4,4)$ provided that $n\ge 5$. So, the claim holds true.
\end{proof}

\begin{lemma}
\label{lemma011}
Let $s$ be a fully normalized binary string of length $n$ such that $n\ge 5$ and we always perform prefix block-interchange operations, then $d_{pbi}^g(s) = \lceil \frac{n-2}{3} \rceil$.
\end{lemma}
\begin{proof} 
As we always get a 3-pblockInterchange when $n\ge 5$ (lemma \ref{lemma0}), then the length of the string will be reduced by 3 after each operation. Thus, the bound holds true.
 \end{proof}

\subsection{Grouping Binary Strings by Restricted Prefix Block-interchanges}
\label{restrictedpblock}
Now, let us consider a different scenario where we always keep the first block of consecutive
 identical symbol(s) of a string fixed, i.e., we keep the first symbol of a normalized binary string fixed. So, while performing a restricted prefix block-interchange operation we actually mean $\beta(2,2,x,y)$. We refer to this operation as the restricted prefix block-interchange operation.  The motivation for this apparently strange setting comes from the existence of circular genomes. While handling a circular permutation (genome), we need to keep the first position fixed. There exist some works where such circular permutations are considered \cite{mixtacki,solomon}.


\begin{property}
\label{property3}
Let $s$ be a fully binary normalized string of length $n$, where $2 < x \le y <n$. If we have a block $s[2]$ and another block $s[x]\ldots s[y]$ such that $s[2]=s[y+1]$, $s[2] =s[x-1]$, $s[1]=s[x]$ and $s[3]=s[y]$, then we get a 4-pblockInterchange using the restricted prefix block-interchange operations.
\end{property}

For example, let $s=10101010$. Since, $s[2] = s[5+1]$, $s[2]=s[5-1]$, $s[1]=s[5]$ and $s[3]=s[5]$ (i. e., $x=5$ and $y=5$), $\beta(2,2,5,5)$ can be applied on $s$ as follows:  $1 \overline{0}10\underline{1}010=\textbf{111000} 10 = \textbf{10} 10 = 1010$.

\begin{lemma}\label{lemma_binary1}
Let $s$ be a fully normalized binary string of length $n$ such that $n\ge 7$. Then, we always get 4-pblockInterchange(s) using restricted prefix block-interchange operations.
\end{lemma}

\begin{proof} 
As $s$ is a fully normalized binary string, there are only four kinds of such strings: $1010\ldots 10$ of even length, $0101\ldots 01$ of even length, $1010\ldots 101$ of odd length and $0101\ldots 010$ of odd length. So, for all these strings, we always get that the first symbol is same as the fifth symbol, the second symbol is same as the fourth symbol, the fourth symbol is same as the sixth symbol and the third symbol is same as the seventh symbol. Then, according to Property \ref{property3}, we always get a 4-pblockInterchange for each of the four possible combinations (such that $n\ge7$) which reduces the length of the string by 4. 
 \end{proof}

\begin{theorem}\label{thm_binary1}
Let $s$ be a fully binary normalized string of length $n$ such that $n\ge 2$ and we always perform restricted prefix block-interchange operations, then $d_{rpbi}^g(s) = \lceil \frac{n-2}{4} \rceil$, where $d_{rpbi}^g(s)$ is the number of restricted prefix block-interchanges required to group all symbols of the string $s$.
\end{theorem}
\begin{proof}
First, we prove the bound for all fully normalized binary strings where $n<7$, then we prove for $n\ge 7$. When $n<7$, we need one prefix block interchange operation to group all symbols in $s$. For this case, length of $s$ can be 3, 4, 5 or 6. We can easily group these strings as follow:
\begin{itemize}
\item $101\xrightarrow{\beta(2,2,3,3)} \textbf{11}0 = 10$
\item $1010\xrightarrow{\beta(2,2,3,3)} \textbf{1100} = 10$
\item $10101 \xrightarrow{\beta(2,2,5,5)} \textbf{11100} = 10$ 
\item $101010\xrightarrow{\beta(2,2,5,5)} \textbf{111000} = 10$. 
\end{itemize}
Note that 1010 and 0101 are in fact similar strings because we get the latter from the former by relabeling `0' to `1' and `1' to `0' as in \cite{amit,khaled}. So, we need one restricted prefix block-interchange to group 010, 0101, 01010 and 010101 which satisfies the proposed bound. Now, when $n \ge 7$, we always get 4-pblockInterchange(s) to perform on $s$ (Lemma \ref{lemma_binary1}), which will reduce the length of $s$ by 4. Thus the proposed bound always holds true.
 \end{proof}
Note that our deduced upper bound to group fully binary normalized string for restricted prefix block-interchanges is better than other results presented in Table \ref{summary-tab}. After grouping is done, we get either 01 or 10, where 01 is already sorted. So, if we get 10 then we treat it as a permutation and need one extra prefix block-interchange operation to sort it. Thus, $d_{pbi}^s(s) = \lceil \frac{n-2}{3} \rceil + 1$, when we need to sort 10, otherwise $d_{pbi}^s(s) =d_{pbi}^g(s) = \lceil \frac{n-2}{3} \rceil$. 

We also present Algorithm \ref{algo1} to compute the restricted prefix block-interchange distance to group fully normalized binary string(s). In Algorithm \ref{algo1}, \emph{swap} and \emph{normalize} operations take constant time \cite{compete}. We can perform the \emph{swap} operation by exchanging one or two symbols with one or two symbols respectively. Similarly, \emph{normalize} operation can be done always by checking first to seventh symbols in $s$. So, the \emph{while} loop will execute not more than $n$ times. Clearly, the time complexity of Algorithm \ref{algo1} is $O(n)$.
\begin{theorem}\label{thm_binary2}
The problem of grouping fully normalized binary strings by restricted prefix block-interchanges can be solved by a linear-time algorithm.
\end{theorem}

\begin{algorithm}[H]

\SetAlgoLined
\KwIn{$s$ is a fully binary normalized string, where length of $s$, $n \ge 3$}
\BlankLine
$count = 0$\;

\While{($n \ge 7$)}
{
	swap($s[2]s[3]$, $s[5]s[6]$)\;
	normalize($s$)\;
	$n=n-4$\;
	$count=count+1$\;
	
}
swap($s[2]$, $s[n-1]$)\;
$count=count+1$\;			
normalize($s$)\;
$d_{rpbi}^g(s)=count$\;
$return$ $d_{rpbi}^g(s)$ \;
\caption{ Algorithm to group fully normalized binary string by restricted prefix block-interchanges}
\label{algo1}
\end{algorithm}

\section{Grouping and Sorting Ternary Strings}\label{section_grpt}
In this section, we deduce both the grouping and sorting distances for ternary strings. It can be noted that grouping ternary strings is not as simple as grouping binary strings. We start with an easy lemma.

\begin{lemma}\label{thm:always1}
In a fully normalized ternary string of length greater than 3, we can always perform a $1$-pblockInterchange.
\end{lemma}
\begin{proof}
Consider a normalized ternary string $s$ of length $n>3$. So, $s=s[1]\ldots s[l]\ldots s[i]\ldots s[j]\ldots s[n]$, here it is obvious that $1\le l < i \le j < n$. Now, we take $s[1]\ldots s[l]$ as a prefix. As this is a normalized ternary string, somewhere we will find a block $s[i]\ldots s[j]$ such that any one of the following conditions (1-3) is satisfied; and consequently, we get a 1-pblockInterchange.
\begin{itemize}
\item{Condition 1: $s[1]=s[i-1]$}
\item{Condition 2: $s[j]=s[l+1]$}
\item{Condition 3: $s[l]=s[j+1]$}
\end{itemize} 
Since one of the above cases always occurs for fully ternary strings when $n>3$, the result follows.  \end{proof}

Now, we assume that we can perform a $2$-pblockInterchange on a string. Note that, instead of taking a prefix of length at least 2, we can always take a prefix of length 1 (i.e., the first symbol), and perform a 1-pblockInterchange. Thus, the presence of a  2-pblockInterchange always ensures that there is also a 1-pblockInterchange. For example, let's consider 10201. There is a $2$-pblockInterchange: $\overline{1}02\underline{0}1 \Rightarrow \textbf{00}2\textbf{11} \Rightarrow 021$. We could also get a 1-pblockInterchange as follows: $\overline{1}\underline{020}1 \Rightarrow 020 \textbf{11} \Rightarrow 0201$.


The lower bound for the \emph{grouping} of a ternary string remains the same as that of binary strings; but, as can be seen from Theorem~\ref{thm:ternaryBound} below, the upper bound differs. At first, we give an easy but useful lemma.

\begin{lemma}\label{lemma_2Trans}
Suppose $s[1]\ldots s[n]$ is a fully normalized ternary string. If we have a prefix $s[1]\ldots s[i], 1 \leq i<j \leq k < n$ such that $s[1] = s[j-1]$ and $s[i] = s[k+1]$, then we have a 2-pblockInterchange.
\end{lemma}
\begin{proof}
We have a fully normalized ternary string $s$. After performing a prefix block-interchange $\beta(1,i,j,k)$ on $s$, $s[1]$ will be adjacent to $s[j-1]$ and $s[i]$ will be adjacent to $s[k+1]$. Then, we will be able to eliminate one of $s[1]$ or $s[j-1]$ and one of $s[i]$ or $s[k+1]$. This ensures that the length of $s$ will be reduced by 2 which completes the proof. 
\end{proof}

%


\begin{theorem}\label{thm:ternaryBound}
\emph{(Grouping ternary strings applying prefix block-interchanges)}
Let $s$ be a fully normalized ternary string. Then, $d_{pbi}^{g}(s) \leq \lceil \frac{n-2}{2} \rceil $ where $n$ is the length of the string.
\end{theorem}

\begin{proof}
At first, we will prove the bound for all normalized ternary strings where $n<7$. Then, we will describe the proof for all normalized ternary strings where $n\geq 7$. In what follows, we only consider strings starting with $1$. This does not lose the generality since we can always use relabeling for strings starting with $0$ or $2$. 
As strings are fully ternary, we do not need to work with $n\leq 3$. For $n = 4$, we require a 1-pblockInterchange and that is optimal, hence $d_{pbi}^{g}(s) = 1$, and the upper bound is satisfied. We give all fully normalized ternary strings of length 4 starting with 1 in List (\ref{List-Pref}). When $n=5$, we can always perform either 1-pblockInterchange or 2-pblockInterchange and thus $d_{pbi}^{g}(s) \le 2$. Now, we apply prefix block-interchanges for $n=6$, we always get $d_{pbi}^{g}(s) \leq 2$. Supplementary Table \ref{tab1} shows the derivation of the grouping distance for fully normalized ternary strings of length 5 and 6. It is easy to realize that, by Lemma~\ref{thm:always1}, we can always satisfy the given upper bound. Thus the upper bound is proved for $n<7$. 

\begin{equation}\label{List-Pref}
1012, 1010, 1021, 1020, 1201, 1202, 1210 \mathrm{~and~} 1212.
\end{equation}




Now we consider $n\geq7$. Note carefully that for any string starting with $1$, we can only have one of the eight prefixes of length 4 from List (\ref{List-Pref}).


For a string of length $n$, if we can give a 2-pblockInterchange, the resulting reduced string will start with its previous starting symbol and it may be 1, 0 or 2. For the latter two cases (0 and 2), we can again use relabeling as mentioned before. Therefore we can safely assume that the reduced string will have any of the 8 prefixes of List~(\ref{List-Pref}). Hence, it suffices to prove the bound considering each of the prefixes of List~(\ref{List-Pref}). We will show the lists for 7 length fully normalized ternary strings each of which will have a 4 length prefix from List~(\ref{List-Pref}). Then we will provide arguments for each of the prefixes from List~(\ref{List-Pref}). 

Firstly, it is easy to note that the prefixes 1010 and 1212 themselves are binary strings. So, they take one prefix block-interchange operation. Therefore, we can safely exclude one of them from the following discussion. Now, if we relabel $``0"$ to $``2"$ and $``2"$ to $``0"$, then we get 1201 from 1021. Similarly, after relabeling we get 1202 from 1020 and 1210 from 1012. So, we can safely exclude these also from the following discussion. In what follows, when we refer to the prefixes of List~(\ref{List-Pref}), we would actually mean all the prefixes excluding 1201, 1202, 1210 and 1212. 
%
\subsection*{1010}
We first give all possible strings of length 7 having prefix $1010$ in List (\ref{List-Pref1}). The possible 5 length strings having this prefix are $10101$ and $10102$. 

%

\begin{equation}\label{List-Pref1}
1010101, 1010102, 1010120, 1010121, 1010201, 1010202, 1010210 \mathrm{~and~} 1010212.
\end{equation}

In List (\ref{List-Pref1}), all strings will have at least one $2$-pblockInterchange. When we have a prefix symbol $1$, we need to lookup somewhere for a symbol $1$ so that we can perform a $2$-pblockInterchange. Similarly, when we have a prefix block $10$, we need to lookup somewhere for a suffix block $1\ldots 0$. In that case, we will be able to perform a $2$-pblockInterchange. This is also true for longer prefix blocks. 

Let us check this with examples. The string $1010121$ satisfies the bound as follows: $\overline{1}01\underline{0}121 \Rightarrow 0121 \Rightarrow \overline{01}\underline{2}1 \Rightarrow 201$. We see that, for $n=7$, we need only two prefix block-interchanges holding the bound true. However, if we can reduce the string length by at least $1$ after each step, we will also achieve the desired bound. Here, all strings satisfy the bound. Let us see another example for string 1010212. The reduction steps are as follows: $\overline{1}01\underline{0}212 \Rightarrow 01212 \Rightarrow \underline{0}12\overline{1}2 \Rightarrow 1202 \Rightarrow \underline{12}\overline{0}2 \Rightarrow 012$. This takes three steps which satisfies the bound. Now, if we add 0 or 1 after 1010212, then the resulting string will be 10102120 or 10102121 in which case we will be able to perform a $2$-pblockInterchange. So, appending any length of ternary string(s) after 1010212 will also have $2$-pblockInterchange in the reduction steps.


\subsection*{1012}
Here, again we first give the list of 7 length strings having prefix 1012 (see List (\ref{fig2})). Here, the possible 5 length strings of this prefix are $10120$ and $10121$.

\begin{equation}\label{fig2}
1012010, 1012012, 1012020, 1012021, 1012101, 1012102, 1012120 \mathrm{~and~} 1012121.
\end{equation}

In this case, we see that all strings having length 7, will have a $2$-pblockInterchange  as follows. When we have a prefix block $10$, we need to lookup somewhere for a block $1\ldots0$ in the latter part of that string. Similarly, when we have a prefix block $1012$, we need to lookup for a block $1\ldots2$ in the latter part of that string. Then we will be able to perform a $2$-pblockInterchange which will reduce the string length by 2. Then according to Lemma \ref{thm:always1}, we will get at least two consecutive $1-$pblockInterchange which will satisfy the bound. This is also true for higher length prefix blocks.

\subsection*{1020}
We present all strings of 7 length having prefix $1020$ in List (\ref{fig3}). 

\begin{equation}\label{fig3}
1020101, 1020102, 1020120, 1020121, 1020201, 1020202, 1020210 \mathrm{~and~} 1020212.
\end{equation}

Here, according to Lemma \ref{lemma_2Trans}, all strings will have a $2$-pblockInterchange. When we have a prefix block $102$, we need to lookup for a block $1\ldots2$ in the latter part of that string. Similarly, when we have a prefix block $1..0$, we need to lookup for a block $1\ldots 0$ in the latter part of that string so that after applying a prefix block-interchange, we get a $2$-pblockInterchange. Then, according to Lemma \ref{thm:always1}, the bound will be satisfied. Thus, if we append any length of ternary string(s) after 7 length strings having prefix $1020$, that will also have $2$-pblockInterchange.


\subsection*{1021}
We present all strings of 7 length having prefix $1021$ in List (\ref{fig4}). 
\begin{equation}\label{fig4}
1021010, 1021012, 1021020, 1021021, 1021201, 1021202, 1021210 \mathrm{~and~} 1021212.
\end{equation}
If we observe the 7 length strings carefully based on Lemma \ref{lemma_2Trans}, we find that all strings have at least one $2$-pblockInterchange. Then it will be reduced to 5 length ternary strings in one step. According to Lemma \ref{thm:always1}, in the next two consecutive steps, these 5 length ternary strings will be reduced to 3 length strings after performing a $1$-pblockInterchange appropriately in each step. As a result, the bound holds true. Through the reduction steps, if we can perform a single $2$-pblockInterchange, the bound will hold true. Thus, if we append any length of ternary string after these 7 length normalized ternary strings, the bound always holds true.

This completes the proof.
 \end{proof}

\begin{algorithm}[H]

\SetAlgoLined
\KwIn{$s$, a fully normalized ternary string}
\BlankLine
$count \Leftarrow 0$\;
$twoBlockInterchangeDone  \Leftarrow false$\;
\While{$  \mid s\mid >$3}
{
	\For {$i=1; i<\mid s\mid; i=i+1$}	
	{
		take the first symbol of input string \;
		take the $i^{th}$ symbol of the input string \;
		$firstblock \Leftarrow s[1]s[i]$ \;
		check whether this string is a substring of the current suffix \;
		\For{$j=i+1; j< (\mid s\mid-1); j=j+1$}
		{
			\For{$k=j+1;k<(\mid s \mid -1); k=k+1$}
			{
				$secondblock\Leftarrow s[k]s[k+1]$ \;
				\If{$firstblock==secondblock$}
				{
					$s \Leftarrow$ perform a $2$-pblockInterchange \;
					$count  \Leftarrow count + 1$ \;
					$twoBlockInterchangeDone  \Leftarrow true$ \;
					$break$ \;
				}
			}
		}
		\If {$twoBlockInterchangeDone == false$}	
		{			
			$s \Leftarrow$ perform a $1$-pblockInterchange \;
			$count  \Leftarrow count + 1$ \;
   		}	
	}
	$twoBlockInterchangeDone  \Leftarrow false$ \;
}
\caption{Algorithm to group fully normalized ternary string by prefix block-interchanges (\emph{s}:input string)}
\label{alg12}
\end{algorithm}

We present Algorithm \ref{alg12} to group fully normalized ternary strings in polynomial time. Here, $\mid s\mid$ denotes the length of the string $s$. The outer \emph{while} loop iterates until all the symbols of the string $s$ are grouped. We make the string $firstblock$ apppending the symbol of the first position and the symbol of the $i^{th}$ position. Then, we make another string $secondblock$ appending the symbol of the $k^{th}$ position and $(k+1)^{th}$ position such that $i<j<k$. Now, if $firstblock$ is equal to $secondblock$ then we actually get that $s[1]$ is equal $s[k]$ and $s[i]$ is equal to $s[k+1]$ which are basically the conditions of Lemma \ref{lemma_2Trans}. So, according to Lemma \ref{lemma_2Trans}, we will get a $2$-pblockInterchange here and thus the length of the string will be reduced by 2. If a $2$-pblockInterchange is not found, we will perform a $1$-pblockInterchange which is always possible according to Lemma \ref{thm:always1} and thus the length of the string will be reduced by 1. When all symbols are grouped, the length of the string will be 3 for fully ternary normalized string and so the \emph{while} loop will terminate. Clearly, our presented algorithm (Algorithm \ref{alg12}) will take $O(n^4)$ time to execute, where $n$ is the length of the string.

We achieve the same bound as in Theorem \ref{thm:ternaryBound}, if we want to group normalized ternary strings applying restricted prefix block-interchanges. In supplementary Table \ref{tab002}, we show the derivations of the grouping distance for normalized ternary strings of length 5 and 6 applying restricted prefix block-interchanges.

\begin{theorem}
\emph{(Sorting distance for ternary strings)}
Let $s$ be a fully normalized ternary string. Then, an upper bound for sorting ternary string is $d_{pbi}^{s}(s)\leq \lceil \frac{n-2}{2}\rceil + 2$ where $n$ is the length of the string.
\end{theorem}
\begin{proof}
After grouping a ternary string, we can have one of the following grouped strings: $012, 021, 102, 120, 201$ and $210$. Among these, $012$ is already sorted. We need one more prefix $0$-pblockInterchange to sort $210,102,120$ and $201$. We need two more prefix $0$-pblockInterchanges to sort $021$. Hence the result follows.
 \end{proof}

\subsection{Classification}\label{sec:class}
In this section, following the trend of \cite{amit,khaled}, we will briefly identify two classes of fully normalized ternary strings. We have already shown in Theorem \ref{thm:ternaryBound} that it needs at most $\lceil \frac{n-2}{2}\rceil$ prefix block-interchange operations to group fully normalized ternary strings. We will discuss two classes; the first one consists of the fully normalized ternary strings that require $\lceil \frac{n-3}{2}\rceil$ or $\lceil \frac{n-2}{2}\rceil$ prefix block-interchange operations (Class 1) and the other comprising strings that need roughly $\frac{n}{3}+1$ or less operations (Class 2). We are not able to classify all the fully normalized ternary strings; however we provide below a classification of a large set thereof comprising infinite number of strings:
\begin{itemize}
\item All strings of length 4 and 6 satisfy the bound $\lceil \frac{n-3}{2}\rceil$, so we put them in Class 1 (see supplementary Table \ref{tab1}). Ternary strings reduced to any of these by a series of 1-pblockInterchange or 2-pblockInterchange operations will also belong to Class 1. Ternary strings reduced to any of these by a series of 3-pblockInterchange or 4-pblockInterchange will belong to Class 2.

\item $(10)^{+}210$, $(12)^{+}012$, $120(10)^{+}$, $102(12)^{+}$, $102(10)^{+}$ belong to Class 1. We can apply relabeling on these to find strings starting with 0 or 2 and they also belong to Class 1. If a ternary string is reduced to any one of the previous strings by a series of consecutive $2$-pblockInterchange, then that particular string also belongs to Class 1.

\item $12(102)^{+}$, $12(1020)^{+}$, $10(21)^{+}$, $10(12)^{+}$ belong to Class 2. We can apply one 0-pblockInterchange operation on these strings to get $(102)^{+}12$, $(1020)^{+}12$, $(21)^{+}10$, $(12)^{+}10$, respectively. Now we see that prefix part in each of these strings is binary. So, we can apply 3-pblockInterchange or 4-pblockInterchange operations which will roughly require $\frac{n}{3}+1$ or less operations. We can apply relabeling on those to find strings starting with 0 or 2 which belong to Class 2.
\end{itemize}

\section{Conclusion}\label{conclusion}
In this paper, we have discussed grouping and sorting of fully binary and normalized ternary strings applying prefix block interchanges. In particular we have deduced that, for binary strings the grouping distance is $d_{rpbi}^g(s)=\lceil \frac{n-2}{4}\rceil$ using restricted prefix block-interchanges, which is better than existing results on binary strings applying other genome rearrangement operations \cite{amit,hurkens,khaled,compete}. In addition, we have also provided a linear-time algorithm to sort fully normalized binary strings by restricted prefix block interchanges. We have deduced upper bounds applying prefix block-interchanges for grouping and sorting binary normalized strings as $d_{pbi}^s(s) = \lceil \frac{n-2}{3} \rceil + 1$, when we need to sort 10, otherwise $d_{pbi}^s(s) =d_{pbi}^g(s) = \lceil \frac{n-2}{3} \rceil$. We have also deduced upper bounds for both grouping and sorting normalized ternary strings as $d_{pbi}^{g}(s)\leq \lceil \frac{n-2}{2}\rceil$ and $d_{pbi}^{s}(s)\leq \lceil \frac{n-2}{2}\rceil + 2$ respectively. We have also presented a $O(n^4)$ time algorithm to group normalized ternary strings, where $n$ is the length of the string. 





\begin{appendices}

\section*{Supplementary Tables}

\begin{table}[h!]
\centering
 \begin{tabular}{ | c | l | c| }
   \hline
   $n$ & derivation &  $d_{pbi}^{g}(s)$\\ \hline
   \multirow{7}{*}{5} & $10102\xrightarrow{\beta(1,1,4,4)} \textbf{0011}2 = 012$ & 1\\ 
		& $10120\xrightarrow{\beta(1,2,4,4)} 2\textbf{1100} = 210$ & 1\\
		& $10121\xrightarrow{\beta(1,1,4,4)} 20\textbf{111} = 201$ & 1\\
		& $10201\xrightarrow{\beta(1,1,3,4)} 2\textbf{0011} = 201$ & 1\\
 		& $10202\xrightarrow{\beta(1,1,4,4)} \textbf{00}212 = 0212 \xrightarrow{\beta(1,1,4,4)} \textbf{22}10= 210$ & 2\\
		& $10210\xrightarrow{\beta(1,1,5,5)} \textbf{00}2\textbf{11} = 021$ & 1\\
		& $10212\xrightarrow{\beta(1,2,5,5)} \textbf{2211}0 = 210$ & 1\\ \hline
   \multirow{14}{*}{6} & $101012\xrightarrow{\beta(1,1,4,4)} \textbf{00111}2 = 012$ & 1\\
		& $101021\xrightarrow{\beta(1,1,4,5)} 02\textbf{0111} = 0201\xrightarrow{\beta(1,1,4,4)} 12\textbf{00}= 120$ & 2\\   
		& $101201\xrightarrow{\beta(1,1,4,5)} 2\textbf{00111} = 201$ & 1\\ 
		& $101202\xrightarrow{\beta(1,1,4,5)} 2\textbf{0011}2 = 2012\xrightarrow{\beta(1,1,3,3)} 10\textbf{22}=102$ & 2\\
		& $101210\xrightarrow{\beta(1,2,4,5)} 2\textbf{111}\textbf{00} = 210$ & 1\\
		& $101212\xrightarrow{\beta(1,2,4,5)} 2\textbf{111}02 = 2102 \xrightarrow{\beta(1,1,3,3)} 01\textbf{22}=012$ & 2 \\
		& $102010\xrightarrow{\beta(1,1,6,6)} \textbf{00}20\textbf{11} = 0201 \xrightarrow{\beta(1,1,4,4)} 12\textbf{00}= 120$ & 2\\
		& $102012\xrightarrow{\beta(1,2,6,6)} \textbf{22} 0\textbf{11}0= 2010 \xrightarrow{\beta(1,2,3,3)} 12\textbf{00}= 120$ & 2\\
		& $102020\xrightarrow{\beta(1,2,5,5)} \textbf{22}01\textbf{00} = 2010 \xrightarrow{\beta(1,2,3,3)} 12\textbf{00}= 120$ & 2\\
		& $102021\xrightarrow{\beta(1,2,4,4)} 2020\textbf{11} = 20201\xrightarrow{\beta(1,1,3,3)} \textbf{0022}1=021$ & 2 \\
		& $102101\xrightarrow{\beta(1,1,5,5)} \textbf{00}2\textbf{111} = 021$ & 1\\
		& $102102\xrightarrow{\beta(1,1,5,5)} \textbf{00}2\textbf{11}2 = 0212 \xrightarrow{\beta(1,2,3,3)} 10\textbf{22}= 102$ & 2\\
		& $102120\xrightarrow{\beta(1,2,5,5)} \textbf{221100} = 210$ & 1\\
		& $102121\xrightarrow{\beta(1,2,5,5)} \textbf{2211}01 = 2102 \xrightarrow{\beta(1,1,3,3)} 10\textbf{22}=102$ & 2\\ \hline
 \end{tabular}
\caption{(Supplementary) This table shows the derivation of the grouping distance for normalized ternary strings of length 5 and 6 deduced from List (\ref{List-Pref}). $n$ is the length of normalized ternary strings and  $d_{pbi}^{g}(s)$ is the grouping distance. Here, binary strings like 101010 or 121212 are omitted.}
\label{tab1}
\end{table}

\begin{table}[h!]
\centering
 \begin{tabular}{ | c | l | c| }
   \hline
   $n$ & derivation &  $d_{rpbi}^{g}(s)$\\ \hline
   \multirow{6}{*}{5} & $10102\xrightarrow{\beta(2,2,3,4)} \textbf{1100}2 = 102$ & 1\\
		& $10120\xrightarrow{\beta(2,2,3,4)} \textbf{11}2\textbf{00} = 120$ & 1\\
		& $10121\xrightarrow{\beta(2,2,5,5)} \textbf{111}20 = 120$ & 1\\
		& $10201\xrightarrow{\beta(2,2,4,4)} \textbf{11}2\textbf{00} = 120$ & 1\\
 		& $10202\xrightarrow{\beta(2,2,5,5)} 1\textbf{2200}= 120$ & 1\\
		& $10210\xrightarrow{\beta(2,2,4,4)} \textbf{11}2\textbf{00} = 120$ & 1\\
		& $10212\xrightarrow{\beta(2,2,4,5)} \textbf{1122}0 = 120$ & 1\\ \hline
   \multirow{12}{*}{6} & $101012\xrightarrow{\beta(2,2,5,5)} \textbf{111}\textbf{00}2 = 102$ & 1\\ 
		& $101021\xrightarrow{\beta(2,2,6,6)} \textbf{111}02\textbf{0} = 1020\xrightarrow{\beta(2,2,3,3)} 12\textbf{00}=120$ & 2\\ 
		& $101201\xrightarrow{\beta(2,2,6,6)} \textbf{111}2\textbf{00} = 120$ & 1\\ 
		& $101202\xrightarrow{\beta(2,2,3,5)} \textbf{11}2\textbf{00}2 = 1202\xrightarrow{\beta(2,2,3,3)} 10\textbf{22}=102$ & 2\\
		& $101210\xrightarrow{\beta(2,2,5,5)} \textbf{111}2\textbf{00} = 120$ & 1\\
		& $101212\xrightarrow{\beta(2,2,4,4)} \textbf{111}202 = 1202 \xrightarrow{\beta(2,2,3,3)} 10\textbf{22}=102$ & 2 \\
		& $102010\xrightarrow{\beta(2,3,5,6)} \textbf{11000}2 = 102$ & 1\\
		& $102012\xrightarrow{\beta(2,2,5,6)} \textbf{112200}= 120$ & 1\\
		& $102020\xrightarrow{\beta(2,2,5,5)} 1\textbf{22}\textbf{000} = 120 $ & 1\\
		& $102021\xrightarrow{\beta(2,2,5,5)} 1\textbf{2200}1 = 1201\xrightarrow{\beta(2,2,4,4)} \textbf{11}02=102$ & 2 \\
		& $102101\xrightarrow{\beta(2,3,6,6)} \textbf{11100}2 = 102$ & 1\\
		& $102102\xrightarrow{\beta(2,3,4,5)} \textbf{110022}=102$ & 1\\
		& $102120\xrightarrow{\beta(2,2,4,5)} \textbf{112200} = 120$ & 1\\
		& $102121\xrightarrow{\beta(2,3,6,6)} \textbf{111}202 = 1202 \xrightarrow{\beta(2,2,3,3)} 10\textbf{22}=102$ & 2\\ \hline
 \end{tabular}
\caption{(Supplementary) This table shows the derivation of the grouping distance for normalized ternary strings of length 5 and 6 deduced from List (\ref{List-Pref}). $n$ is the length of normalized ternary strings and  $d_{rpbi}^{g}(s)$ is the grouping distance applying restricted prefix block-interchanges. Here, binary strings like 101010 or 121212 are omitted.}
\label{tab002}
\end{table}

\end{appendices}

\end{document}